\documentclass[final,5p,times,twocolumn]{elsarticle}
 
\usepackage{amssymb}
\usepackage{amsmath}
\usepackage{comment}
\usepackage{listings}

\usepackage{mathtools,amsthm}
\newtheoremstyle{case}{}{}{}{}{}{:}{ }{}
\theoremstyle{case}

\usepackage{graphicx}
\usepackage{subfig}
\usepackage{multirow}

\newtheoremstyle{mystyle}
  {}
  {}
  {\itshape}
  {}
  {\bfseries}
  {.}
  { }
  {}

\theoremstyle{mystyle}
\newtheorem{thm}{Theorem}
\newtheorem{lem}{Lemma}

\newdefinition{rmk}{Remark}
\newproof{pf}{Proof}
\newproof{pot}{Proof of Theorem \ref{thm2}}
\usepackage{enumitem}
\newlist{steps}{enumerate}{1}
\setlist[steps, 1]{label = Step \arabic*:}

\theoremstyle{definition}
\newtheorem{definition}{Definition}[section]
\usepackage{booktabs}
\usepackage{tikz}
\usetikzlibrary{shapes.geometric, arrows}

\begin{document}

\begin{frontmatter}

\title{A contraction theory approach to observer-based controller design for glucose regulation in type 1 diabetes with intra-patient variability } 



\author[label1]{Bhabani Shankar Dey}
\address[label1]{Department of Electrical Engineering, Indian Institute of Technology Delhi, Hauz Khas, New Delhi, 110016, India}
\address[label2]{GE Global Research, Bengaluru, Karnataka, 560066, India}
\address[label3]{Department of Electrical  Engineering, Indian Institute of Technology Kanpur, Kalyanpur, Kanpur, Uttar Pradesh, 208016, India}
\cortext[cor2]{Corresponding Author: Anirudh Nath}
\ead{bhabanishankar440@gmail.com}
\author[label2]{Anirudh Nath\corref{cor2}}
\ead{anirudh.nath88@gmail.com}
\author[label3]{Abhilash Patel}
\ead{apatel@iitk.ac.in}
\author[label1]{Indra Narayan Kar}
\ead{ink@ee.iitd.ac.in}

\begin{abstract}
While the Artificial Pancreas is effective in regulating the blood glucose in the safe range of 70-180 mg/dl in type 1 diabetic patients, the high intra-patient variability, as well as exogenous meal disturbances, poses a serious challenge. The existing control algorithms thus require additional safety algorithms and feed-forward actions. Moreover, the unavailability of insulin sensors in Artificial Pancreas makes this task more difficult. In the present work, a subcutaneous model of type 1 diabetes (T1D) is considered for observer-based controller design in the framework of contraction analysis. A variety of realistic multiple-meal scenarios for three virtual T1D patients have been investigated with $\pm30\%$ of parametric variability. The average time spent by the three T1D patients is found to be 77\%, 73\% and 76\%, respectively. A significant reduction in the time spent in hyperglycemia ($>180$ mg/dl) is achieved without any feed-forward action for meal compensation. 
\end{abstract}

\begin{keyword}
Artificial pancreas \sep observer \sep contraction analysis \sep intra-patient variability\sep nonlinear control
\end{keyword}

\end{frontmatter}

\section{Introduction}\label{sec1}
Deficiency of insulin is prominent in type 1 diabetes patients (T1DPs) due to auto-immune destruction of the insulin-secreting $\beta$-cells in the pancreas. It results in prolonged elevated glucose levels ($>180$ mg/dl) in the blood plasma, termed as hyperglycemia \cite{cinar2018artificial}. Thus the T1DPs have to rely upon insulin therapy in terms of multiple daily insulin injections to maintain normalcy in glucose level (70-180 mg/dl). In insulin therapy, dosage of insulin is manually calibrated for which requires the carbohydrate counting apriori \cite{MAJDPOUR2019S8}. Any overestimation or underestimation of carbohydrate counting can lead to inappropriate insulin dosages resulting in hyperglycemic ($>180$ mg/dl) and hypoglycemic ($<70$ mg/dl) episodes. Hyperglycemic and hypoglycemic instances can lead to ineffective glucose management in T1DPs. A large number of health-related complications, both macrovascular and microvascular, occur in T1DPs due to hyperglycemia and hypoglycemia \cite{BJORNSTAD2018809,jdrf}.\par

The issues as mentioned above can be addressed by the Artificial Pancreas (AP). The AP is an externally worn device equipped with a continuous micro-fluid insulin delivery system (insulin pump) and a glucose sensor. The performance of the AP is largely affected by (i) the delay in absorption of the subcutaneous insulin infusion \cite{bondia2018insulin}, (ii) the circadian oscillations in insulin sensitivity \cite{insulin} and (iii) the uncertainty in meal absorption \cite{sensor}. Due to these factors, the T1DPs using AP face the problem of postprandial hyperglycemia and late hypoglycemia \cite{hypo}. The postprandial hyperglycemia refers to the sudden increase in glucose concentration after the meal intake. The late hypoglycemia refers to the sudden decrease in glucose concentration due to excessive insulin infusion at the time of the meal. So the existing APs often rely on additional feed-forward strategy and safety algorithms to achieve effective glucose regulation (70-180 mg/dl). The feedforward strategy requires accurate information about the carbohydrate contents of the meals to determine the insulin dosage required to compensate for the effect of the meal \cite{MARCHETTI2008149}. On the other hand, safety algorithms prevent excessive insulin infusion by estimating the existing insulin concentration in the body \cite{FUSHIMI20181}. Thus, the focus of the current work hovers around achieving a tight glycemic control while minimizing the time spent in postprandial hyperglycemia and avoiding late hypoglycemia. In the current work, a model-based feedback control strategy is proposed that does not require any additional feedforward or safety algorithms. \par

The functionality of an automated AP can be divided into two steps: (i) the task of state estimation and (ii) the task of control design. The existing state estimation techniques and control algorithms for the AP is discussed current and succeeding paragraphs, respectively. A summary of a wide variety of state estimation techniques has been reported in the literature \cite{bondia2018insulin}. The main limiting factors of the success of the existing works on state estimation are intra-patient variability (variation of system parameters that occur inherently in the system) and uncertainty in the meal absorption dynamics \cite{guarcost}. There are several attempts to address these issues as in \cite{guarcost,ijs,aem}. Still, this remains an open problem. The nature of the variability is unknown and, hence, poses some limitations to the current observers where specific assumptions on the uncertainty are involved. The Luenberger observer designs in \cite{hariri2011identification,NATH20197} are based on the nominal system and involve some linear approximations in the design. It can be too restrictive in terms of their applicability. A few robust nonlinear extended Luenberger observers have been designed in some recent studies \cite{guarcost,ijs,aem} for addressing intra-patient variability. But there are some restrictive assumptions on the nature of uncertainty that may be difficult to be characterized in practice. On the other hand, the stochastic filtering approaches motivated by Kalman filtering (KF) techniques and its variants are quite common \cite{7413702}. Algorithms like KFs \cite{kf}, extended KFs \cite{ekf} and unscented KFs \cite{eberle2011unscented,ukfmeal} have been proposed for state estimation in the AP problem.  The requirement of precise information about the model, complex matrix computations, and the difficulty in the characterization of the probability distribution functions one to rethink before implementing them both in theory and practice while keeping in mind the unknown nature of parametric variability in T1DPs. An observer of extended Luenberger structure is proposed in the present work.  The issues related to intra-patient variability and uncertainty in meal disturbance are addressed in the framework of contraction analysis \cite{9067132}. \par 
The existing control algorithms for the AP is discussed from time to time in terms of several survey papers \cite{NATH2018289,blauw2016review,review}. The main objectives of the current work are (i) to minimize the instances of postprandial hyperglycemia, and (ii) to avoid late hypoglycemic instances, in the presence of intra-patient variability.  Mainly, two variants of robust controllers, namely the $H_{\infty}$ control \cite{RuizVelzquez2004BloodGC,mandal2014lmi,6851161} and the sliding mode control (SMC) \cite{AHMAD2017200,GALLARDOHERNANDEZ2013747} exist in the literature. The necessity of structural characterization of the bounded uncertainty and computational complexity makes the practical implementation of the $H_{\infty}$ filter based controllers challenging. The chattering issues that are inevitable in SMC  may trigger hypoglycemia during large intra-patient variability due to excessive insulin infusion \cite{7902222}. Apart from these, most of the control algorithms are augmented with additional safety layers \cite{blauw2016review,6463437,8025565,7870660} and meal compensation techniques (feed-forward action) \cite{MARCHETTI2008149}. It complicates the practical applications of such a multi-layered control approach \cite{8263478}. Thus, an attempt is made in this current work to devise a simple feedback control law based on the theory of contraction analysis, which would be sufficient in itself to tackle these issues. \par

Contraction analysis offers an alternative tool for stability analysis of nonlinear systems, which has got lots of attraction recently \cite{contraction}. The inherent feature of forgetting initial conditions exponentially makes it intriguing among control engineers to apply its principle in this problem \cite{LOHMILLER1998683}. Contraction theory is used in applications like observer design in \cite{sharma2008design, rayguru2019high}. As compared to the stability analysis of uncertain nonlinear systems in the framework of Lyapunov stability, contraction analysis offers a more relaxed framework for the same. For instance, high uncertainty and time variability in the physiological parameters like insulin sensitivity, parameters related to insulin absorption and meal absorption, etc. that exist in T1DPs. As a result, there exists a complicated shifting of the equilibrium concerning the variability in system parameters. This may complicate the Lyapunov stability analysis, and getting a closed-form expression of equilibria in respect of parameters may not be possible \cite{AYLWARD20082163}. This problem can be circumvented by using contraction analysis. \par

In the present work, a nonlinear observer-based design of the control algorithm is proposed for the APS that considers the IVP model in \cite{ivp}. The main highlights of the proposed state estimation based control technique is provided as follows:
\begin{enumerate}
\item[(i)] A nonlinear observer is designed for estimating the glucose and insulin concentrations based on the contraction theory approach.The unknown observer gains are computed by solving a set of linear inequalities.\vspace{-0.2cm}
\item[(ii)] The information about the estimated state variables are exploited to design a feedback control law for the glucose regulation under variability in insulin sensitivity and uncertainty in insulin absorption in the subcutaneous compartment.\vspace{-0.2cm}
\item[(iii)] In comparison to the works in \cite{TURKSOY2017159,SALAMIRA201968}, the proposed control algorithm involves the estimation of insulin concentration in the presence of uncertainty in parameters. Hence it is not required to augment any additional safety layers \cite{6463437} like insulin-on-board (IOB) \cite{FUSHIMI20181} to avoid post-prandial hyperglycemia \cite{TURKSOY2017159}. It also avoids any type of additional feed-forward compensating action for disturbance rejection \cite{MARCHETTI2008149}. \vspace{-0.2cm} 
\item[(iv)] The proposed strategy allows us to design the observer and the controller separately. The overall closed loop system, including controller and observer together, is shown to be stable using the concept of partial contraction theory. \vspace{-0.2cm}
\end{enumerate}

The remaining part of this work is presented in several sections as follows. Section 2 presents the methodology of the current work. It further consists of four subsections that present the mathematical model, theoretical background, observer design, and controller design. Section 3 contains the closed-loop studies' results to show the effectiveness of the theoretical results. In Section 4, the summary of the work is briefly provided.

\section{Methods}
In this section, the dynamical model of the glucoregulatory system of T1DP is presented in the first subsection. In the second subsection, a brief background on contraction analysis is included. The third and fourth subsections, the details about the observer and controller design methodologies are provided. 

\subsection{Mathematical model}
\label{sec2}
A model plays an important role in the effectiveness of the closed-loop control strategy. A large number of mathematical models are available in the existing literature that models the dynamics of the glucose-insulin interactions in T1DPs with varying levels of abstractions and utility, as mentioned in \cite{lio2019automated}. The complexly of these models, in terms of a large number of state variables and complicated nonlinear functions, limit their relevance to control applications \cite{RESALAT20177756}. Recently, a few variants of control-oriented subcutaneous models of T1DPs have been reported in the literature \cite{NATH202094,model,ricardo,ivp}. The model in \cite{ivp} is adopted here due to its structural simplicity and identifiability. It is very important, especially in the context of personalising a control technique for individual patients. \par

\subsubsection{State space representation}
The Medtronic Virtual Patient (MVP) \cite{model} models the glucose-insulin dynamics of type 1 diabetes in terms of a coupled ordinary differential equations (ODE) as
\begin{equation}\label{sys}
\begin{array}{lll}
\dot{x}_1&=&-p_1x_1-x_1x_2+EGP+ R_{a}(t)\\ 
\dot{x}_2&=&-p_2x_2+p_3x_3\\
\dot{x}_3&=&-p_4x_3+p_4x_4\\
\dot{x}_4&=&-p_5x_4+p_6u(t)
\end{array}
\end{equation}
where the state vector $x=[x_1~x_2~x_3~x_4]^T$ represents the blood glucose concentration (mg/dl), the effective insulin in the blood ($min^{-1}$), the plasma insulin concentration (mU/l) and the subcutaneous insulin concentration (mU/l), respectively. The physiological parameter, $p_1$ denotes the glucose effectiveness, $\frac{p_3}{p_2}$ represents the insulin sensitivity factor, $p_4$ represents the time constant of the plasma insulin, $p_5$ denotes the time constant of the subcutaneous insulin and $EGP$ stands for the endogenous glucose production. The values of the parameters are adopted from \cite{model} as provided in Table \ref{tab1}. The the blood glucose concentration is the output, $y$ of the system \eqref{sys}, as mentioned below.
\begin{equation}\label{op}
    y=Cx=x_1
\end{equation}
where $C=[1~0~0~0]$ is the output matrix.

\begin{table*}[h!]
\centering
\caption{Estimated parameters for different subjects \cite{ivp}.}
\label{tab1}
 \begin{tabular}{|c| c| c| c| c| c|c|c|} 
 \hline
 Subjects & $p_1$ & $p_2$ & $p_3$ & $p_4$ & $p_5$ & $EGP$ & $p_6$ \\ [0.5ex]  \hline
1 & $2.20\times 10^{-3}$ & $1.06\times 10^{-2}$ & $8.60\times 10^{-6}$ & 0.0213 & 0.0204 & 1.33 & $1.02\times 10^{-5}$\\
3 & $3.50\times 10^{-3}$ & $2.33\times 10^{-2}$ & $1.079\times 10^{-5}$ & 0.0143 & 0.0141 & 1.07 & $1.55\times 10^{-5}$\\
5 & $4.33\times 10^{-3}$ & $9.63\times 10^{-3}$ & $1.974\times 10^{-6}$ & 0.0217 & 0.0217 & 0.6 & $1.416\times 10^{-5}$\\
\hline
 \end{tabular}
\end{table*}

The equilibrium state of the system in \eqref{sys} is given as, $[x_1~x_2~x_3~x_4]^T=[EGP/p_1~0~0~0]^T$. To shift this non-zero equilibrium to the origin, a translation operation is performed in the original state space. Thus, the resulting deviated state can be expressed as
\begin{equation}\label{transform}
    [x_{1d}~x_{2d}~x_{3d}~x_{4d}]^T=[x_1~x_2~x_3~x_4]^T-[EGP/p_1~0~0~0]^T
\end{equation}
where $(.)^T$ represents the transpose operator. The corresponding deviated dynamics of the system \eqref{sys} can be obtained as
\begin{eqnarray}\label{dev_sys}
\begin{array}{lll}
\dot{x}_{1d}&=&-p_1x_{1d}-\frac{EGP}{p_1}x_{2d}-x_{1d}x_{2d}+ R_{a}(t)\\ 
\dot{x}_{2d}&=&-p_2x_{2d}+p_3x_{3d}\\
\dot{x}_{3d}&=&-p_4x_{3d}+p_4x_{4d}\\
\dot{x}_{4d}&=&-p_5x_{4d}+p_6u(t).
\end{array}
\end{eqnarray}
The corresponding output equation is expressed as
\begin{equation}
    y=Cx=x_{1d}
\end{equation}

\subsubsection{Meal disturbance model}
The meal disturbance model represents the dynamics of the glucose absorption in the gut and its appearance in the blood circulation following the meal intake. As mentioned in \cite{doi:10.3109/14639239709089833}, the dynamics can be modelled as a two-compartmental model, as provided below

\begin{eqnarray}
\begin{array}{lll}
\dot{d}_{1}&=&-\frac{d_1}{t_{max}}+ Bio \times Carb(t)\\ 
\dot{d}_{2}&=&\frac{d_1}{t_{max}}-\frac{d_2}{t_{max}}\\
     R_a(t)&=&\frac{d_2}{t_{max}}
\end{array}
\end{eqnarray}
where $d_1$ and $d_2$ are the amount of glucose in the first and second compartments (mg/dl), $t_{max}$ denotes the time-to-maximum rate of appearance of glucose in the blood (min), t is the time of meal intake (min), Carb(t) denotes the ingested amount of carbohydrates at time t (mg/dl/min), Bio is carbohydrate bioavailability of the meal (unitless) and
$R_a(t)$ represents glucose absorption rate (mg/dl/min) in the gut. The values of the parameters are adopted from \cite{doi:10.3109/14639239709089833} as $t_{max}=43$ min and $Bio=71$.

\subsection{Preliminaries on Contraction Theory}
Contraction theory helps to characterise the temporal behaviour of trajectories for a dynamical system to each other \cite{contractionmainpaper}. The system is said to be contracting if two arbitrary trajectories converge towards each other, forgetting their initial conditions. Fig. \ref{cont} depicts a pictorial representation of two arbitrary trajectories converging towards each other. The use of contraction theory relaxes the prior knowledge of equilibrium and ensures inherent robustness to disturbance.
Here, we briefly present results from contraction theory, which will be used in subsequent sections.

\begin{figure}[!h]
		\centering
		\includegraphics[width=1.0\linewidth]{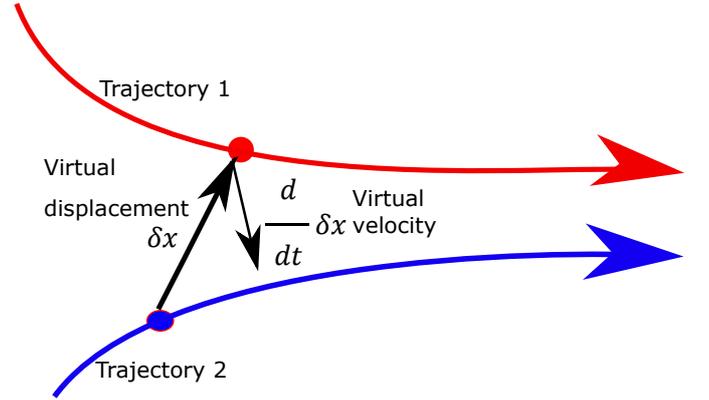}
		\caption{Two trajectories contracting to each other in a contracting region.}
		\label{cont}\vspace{-0.4cm}
\end{figure}

Consider a nonlinear dynamical system,
\begin{equation}\label{eq:nosy}
    \dot{x}=f(x),
\end{equation}
where $x\in\mathbf{R}^n$ is the system states, $f(x)$ is the drift function $f:\mathbf{R}^n\rightarrow\mathbf{R}^n$. The differential dynamics of the system (\ref{eq:nosy}) is,
\begin{equation}\label{eq:difsy}
    {\delta \dot x}=\frac{\partial f}{\partial x}\delta x=J(x)\delta x,
\end{equation}
where $\delta x$ is the virtual displacement of infinitesimal perturbation in $x$, $J(x)$ is the Jacobian matrix of $f(x)$.
The evolution of virtual displacement can be inferred using (\ref{eq:difsy}) as,
\begin{equation}
    \frac{d}{dt}(\delta x^T\delta x)=2\delta x^T\frac{\partial f}{\partial x} \delta x \leq 2 \lambda_{max}(J)\delta x^T\delta x
\end{equation}
where $\lambda_{max}(J)$ is the largest eigenvalue of the Jacobian $J$.
If $\lambda_{max}(J)$ is strictly uniformly negative, then any infinitesimal
length $|\delta x|$ converges exponentially to zero. Hence, all system trajectories of (\ref{eq:nosy}) will converge towards each other, forgetting their initial conditions, as illustrated in Fig. \ref{cont}. If $\lambda_{max}(J)$ is negative in a set, such is called an invariant set, and the region is called a contraction region. In this research article to prove this negative definiteness, we have used the contraction analysis based on matrix measure, which is alternatively known as logarithmic norm. Matrix measure, often referred to as $\mu{(J)}$ is defined as the directional derivative associated with the induced matrix norm evaluated at identity matrix in the direction of matrix $(J)$ \cite{vidyasagar1978matrix}. It should be noted that, if the vector norm is the standard Euclidean norm, then induced matrix measure is the largest eigenvalue of the symmetric part of $(J)$. Based on the different vector norms, different matrix measures can be referred to from Table \ref{table_matrix_measure}.
\begin{table}[h!]
\centering
\caption{Matrix measure of real matrix $A=[a_{ij}]$ corresponding to different vector norms \cite{sontagopenproblems}.}.
\label{table_matrix_measure}
 \begin{tabular}{|c| c|} 
 \hline
 Vector Norms $||(.)||$ & Corresponding Matrix Measure, $\mu(A)$\\ [0.5ex]  \hline
$ ||x||_{1}=\sum_{i=1}^{n}{|x_{i}|}$ & $\mu_{1}(A)=\max_{j}{(a_{jj}}+\sum_{i\neq j}{|a_{ij}|})$\\
\hline
$ ||x||_{2}=(\sum_{i=1}^{n}{|x_{i}|^{2}})^\frac{1}{2}$ & $\mu_{2}(A)=\lambda_{max}(\frac{A+A^{T}}{2})$\\
\hline
$ ||x||_{\infty}=\max_{1\leq i \leq n}{|x_{i}|}$ & $\mu_{\infty}(A)=\max_{i}{(a_{ii}}+\sum_{i\neq j}{|a_{ij}|})$\\\hline
 \end{tabular}
\end{table}
\begin{definition}
Uniform negative definiteness of Jacobian $J(x)$ means there exist a positive constant $\beta$ such that $(J(x)^T+J(x))<-\beta I<0$ in a region $\Omega\subset\mathbf{R}^n$ for $\forall t>0$.
\end{definition}

A relaxed form as partial contraction analysis extends contraction analysis for convergence to a specific behavior. It has been applied to the synchronization of oscillators~\cite{LOHMILLER1998683} and observer design~\cite{observer}.\\
\\
\begin{lem}\label{lemma1}
Consider a nonlinear system of the form
$$\dot{x}=f(x,x)$$
and assume that the auxiliary system (a.k.a virtual system)
$$\dot{y}=f(y,x)$$
is contracting to $y$. If a particular solution of the auxiliary y-system verifies a smooth
the specific property, then all trajectories of the original x-system verify this property exponentially.
The original system is said to be partially contracting.
\end{lem}

The proof of this lemma is presented in ~\cite{partialcontraction}. If the virtual system admits two particular solutions and contracting in nature, then two particular solutions will converge towards each other. The design understanding should be to design a virtual system that will have estimator dynamics as a particular solution and system dynamics as another particular solution. The salient property of contraction analysis is that it can quantify the robustness of an external perturbation.

\begin{lem}\label{lemma2}
Consider that the nominal system $$\dot{x}=f(x)$$ is  contracting and the perturbed model $$\dot{x}_p=f(x_p)+d(t),$$ where ${x}_p$ is the perturbed state, $d(t)$ be a vanishing perturbation satisfying $|d(t)|\leq c_1e^{-c_2t}$ for some $c_1,c_2 > 0$ and $t \geq 0$,
Then, there exist constants $k_1$, $k_2 > 0$ such that the following property holds $\forall t>0$
\begin{equation}
    |x(t)-x_p(t)|< e^{-k_1t}(k_2+|x_0-z_0|)
\end{equation}
\end{lem}

\subsection{Observer Design for Artificial Pancreas}\label{observer}
As discussed in Section 1, the state estimation algorithm is necessary to retrieve information about insulin concentrations in the body, which is of paramount importance in the case of a practical APS.  A nonlinear observer of the extended Luenberger structure is considered for estimating the state variables concerning glucose and insulin concentrations in the plasma and subcutaneous compartments. The dynamical equations of the observer for the system in \eqref{sys} are provided below
\begin{equation}\label{ob}
\begin{array}{lll}
\dot{\hat{x}}_1&=&-p_1\hat{x}_1-\hat{x}_2{x}_1+EGP+l_1(\hat{x}_1-x_1)\\
\dot{\hat{x}}_2&=&-p_2 \hat{x}_2+p_3\hat{x}_3+l_2(\hat{x}_1-x_1)\\
\dot{\hat{x}}_3&=&-p_4 \hat{x}_3+p_4\hat{x}_4+l_3(\hat{x}_1-x_1)\\
\dot{\hat{x}}_4&=&-p_5 \hat{x}_4+u(t)+l_4(\hat{x}_1-x_1)
\end{array}
\end{equation}
where $\hat{x}_i,i=1,...,4$ are the estimated system states and $L=[l_1\ l_2\ l_3\ l_4]^T$ is the unknown observer gain needs to be selected to ensure the convergence of estimated state, $\hat{x}_i,i=1,\dots,4$ in \eqref{ob} to the true states, $x_i,i=1,\dots,4$ in \eqref{sys}. The existence and computation of the observer gain, $L$ is stated in the form of following theorem. 

\begin{thm}\label{theorem-1}
Consider the system dynamics in \eqref{sys} and the observer in \eqref{ob} with the observer gain, $L=[l_1~l_2~l_3~l_4]$. The estimated states, $\hat x_i$ in \eqref{ob} will converge to the true states, $x_i$ in \eqref{sys}, exponentially from arbitrary initial conditions if the observer gain $L$ satisfies the following linear inequalities

\begin{equation}\label{ob_ine}
\begin{array}{lr}
-p_{1}+l_1+l_2+l_3+l_4<0 & -p_{2}+x_1<0 \\
-p_{4}+p_{3}<0 & -p_{5}+p_{4}<0.
\end{array}
\end{equation}

\end{thm}
 
\begin{proof}
Let us consider a virtual system as
\begin{equation}
\begin{array}{lll}\label{virtual}
\dot{s}_1&=&-p_1s_1-x_1{s}_2+EGP+l_1(s_1-x_1)\\
\dot{s}_2&=&-p_2 s_2+p_3s_3+l_2(s_1-x_1)\\
\dot{s}_3&=&-p_4 s_3+p_4s_4+l_3(s_1-x_1)\\
\dot{s}_4&=&-p_5 s_4+u(t)+l_4(s_1-x_1)
\end{array}
\end{equation}
where $s_i,i=1,...,4$ are the states of virtual system. It has two particular solutions. For $s_i=x_i$, the virtual system represents the system dynamics in \eqref{sys}, and for $s_i=\hat{x}_i$, the virtual system represents the observer dynamics in \eqref{ob}. Now, the corresponding differential dynamics of the virtual system in \eqref{virtual} can be obtained as\\
\begin{equation}
  \left[\begin{array}{c}
\delta \dot{s}_{1} \\
\delta \dot{s}_{2} \\
\delta \dot{s}_{3} \\
\delta \dot{s}_{4}
\end{array}\right]=\left[\begin{array}{cccc}
-p_{1}+l_{1} & -x_{1} & 0 & 0 \\
l_{2} & -p_{2} & p_{3} & 0 \\
l_{3} & 0 & -p_{4} & p_{4} \\
l_{4} & 0 & 0 & -p_{5}
\end{array}\right]\left[\begin{array}{c}
\delta s_{1} \\
\delta s_{2} \\
\delta s_{3} \\
\delta s_{4}
\end{array}\right]\\  
\end{equation}

In the compact form, the differential dynamics becomes

\begin{equation}\label{diff_obs}
    \delta\dot{s}=J\delta s
\end{equation}


 If the matrix measure of $J$, in \eqref{diff_obs} is  negative, then the virtual system dynamics in \eqref{virtual} will be contracting. So, one needs to choose the observer gains, $L=[l_1\ l_2\ l_3\ l_4]^T$ in such a way that $\mu_{1}(.)<0$. (Here, the matrix measure is based on $1-norm$). Referring to Table \ref{table_matrix_measure}, the conditions for the observer gains can be computed as in \eqref{ob_ine}. If the inequalities in \eqref{ob_ine} are satisfied, then the virtual system in \eqref{virtual} will be a contracting system. Hence, by Lemma \ref{lemma1}, it can be inferred that the estimated states, $\hat x_i$ in \eqref{ob} converges to the corresponding true states, $x_i$ in \eqref{sys} exponentially.
\end{proof}
 From the physiological knowledge, it is well known that state $x_1$ (blood glucose concentration) evolve in a compact set $X$ \cite{NATH2018289}. Thus, the compact set's boundary values can be substituted as an argument to the inequality in \eqref{ob_ine} and tune the observer gain to satisfy the inequality. 

\begin{rmk}
It is important to note that the inequalities in \eqref{ob_ine} involve only the information of blood glucose concentration, $x_1$, which is the measurable quantity in AP. It provides less conservatism than the observer designs in \cite{guarcost,ijs,aem}. This is because of the requirement of information on the bounds of all the state variables. Apart from blood glucose concentration, it is difficult to get the bounds of other state variables precisely. 
\end{rmk}

\begin{rmk}
The contraction theory provides us with an alternative approach for observer design in contrast to the Lyapunov stability theory. Furthermore, the observer designs in \cite{guarcost,ijs,aem} are based on quadratic Lyapunov functions and involve linear approximations of the nonlinearity, such as Lipschitz condition \cite{aem} and one-sided quasi-Lipschitz condition. Unlike these methods, the proposed observer in \eqref{ob} is less conservative since it does not require any of the above mentioned assumptions. 
\end{rmk}

\begin{figure*}[h!]
\centering
\subfloat[]{
  \includegraphics[width=70mm,height=3cm]{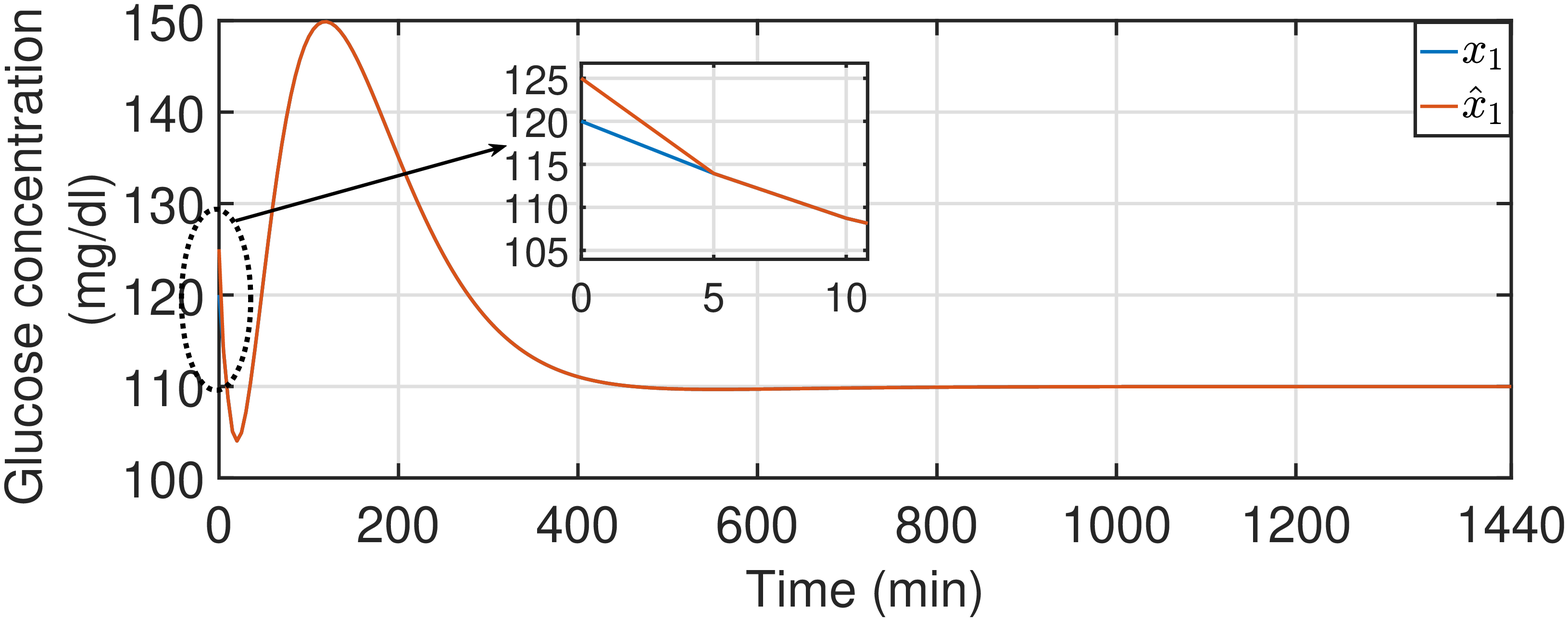}
}
\subfloat[]{
  \includegraphics[width=70mm,height=3cm]{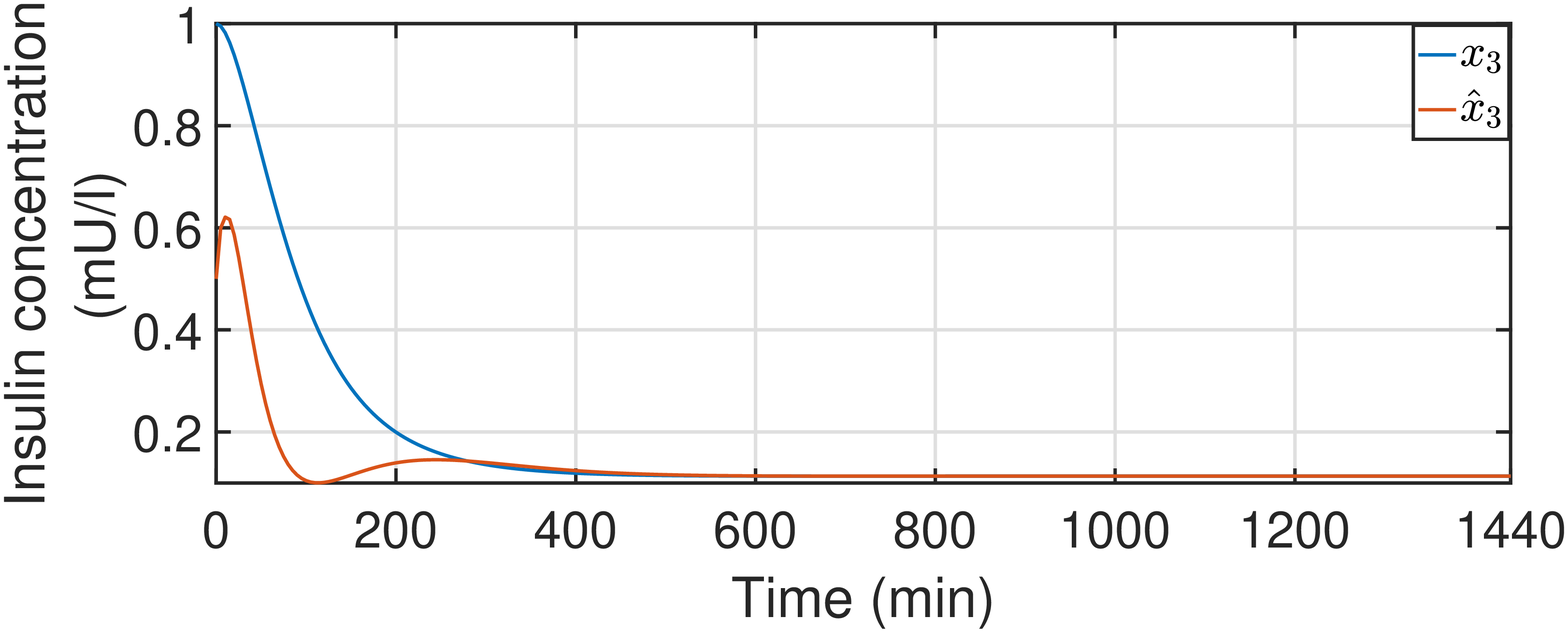}
}
\\
\subfloat[]{
  \includegraphics[width=70mm,height=3cm]{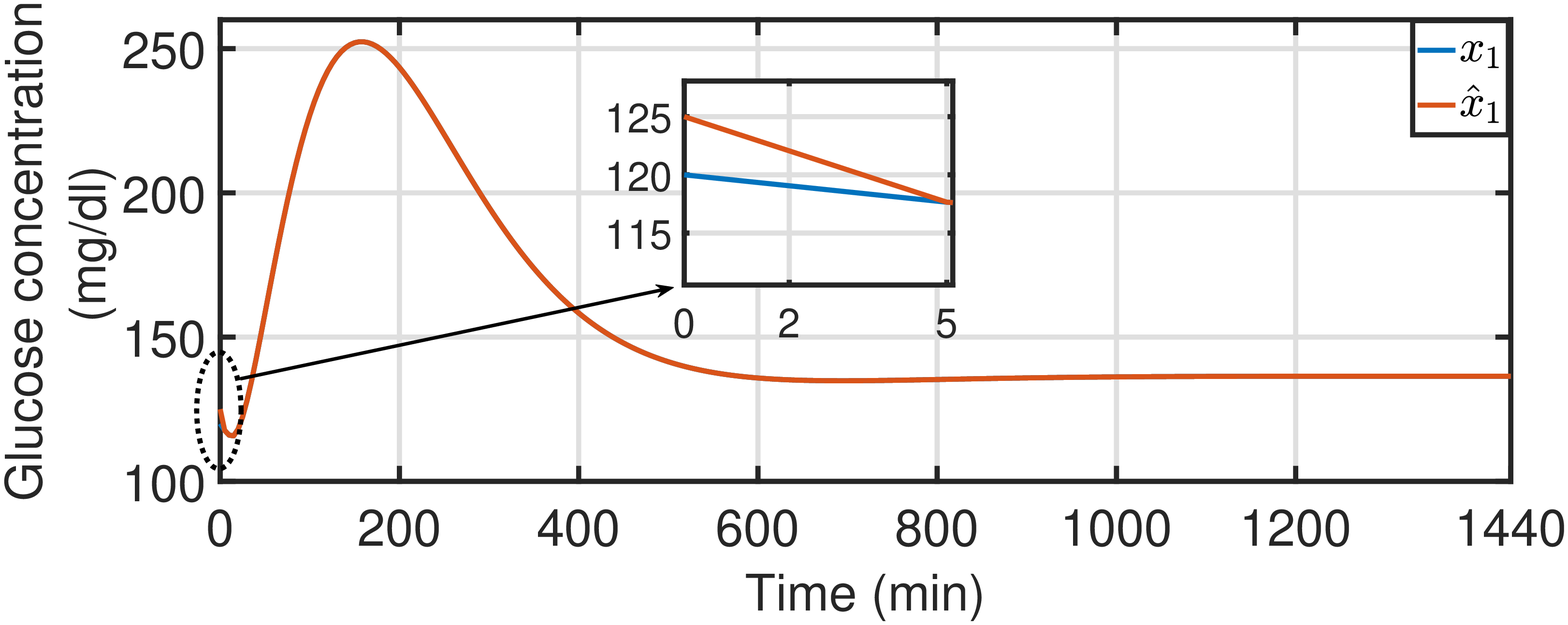}
}
\subfloat[]{
  \includegraphics[width=70mm,height=3cm]{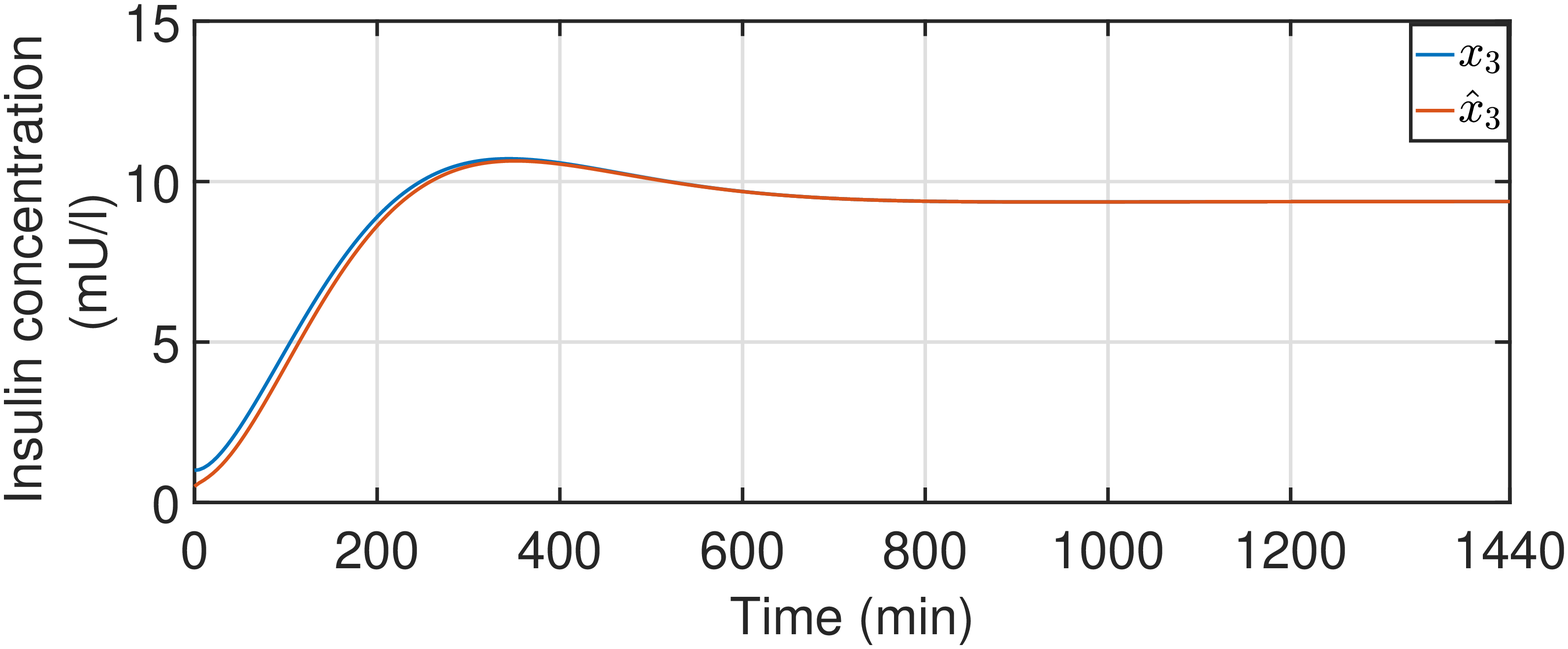}
}
\\
\subfloat[]{
  \includegraphics[width=70mm,height=3cm]{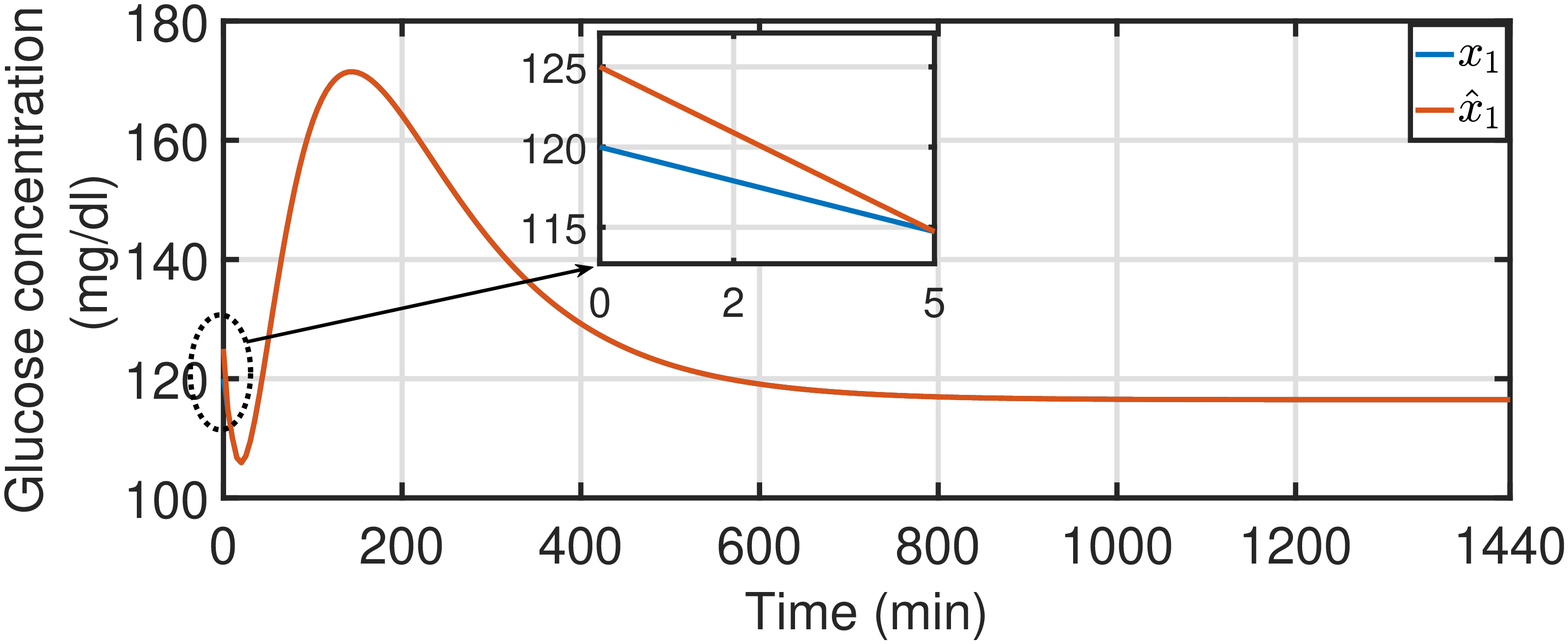}
}
\subfloat[]{
  \includegraphics[width=70mm,height=3cm]{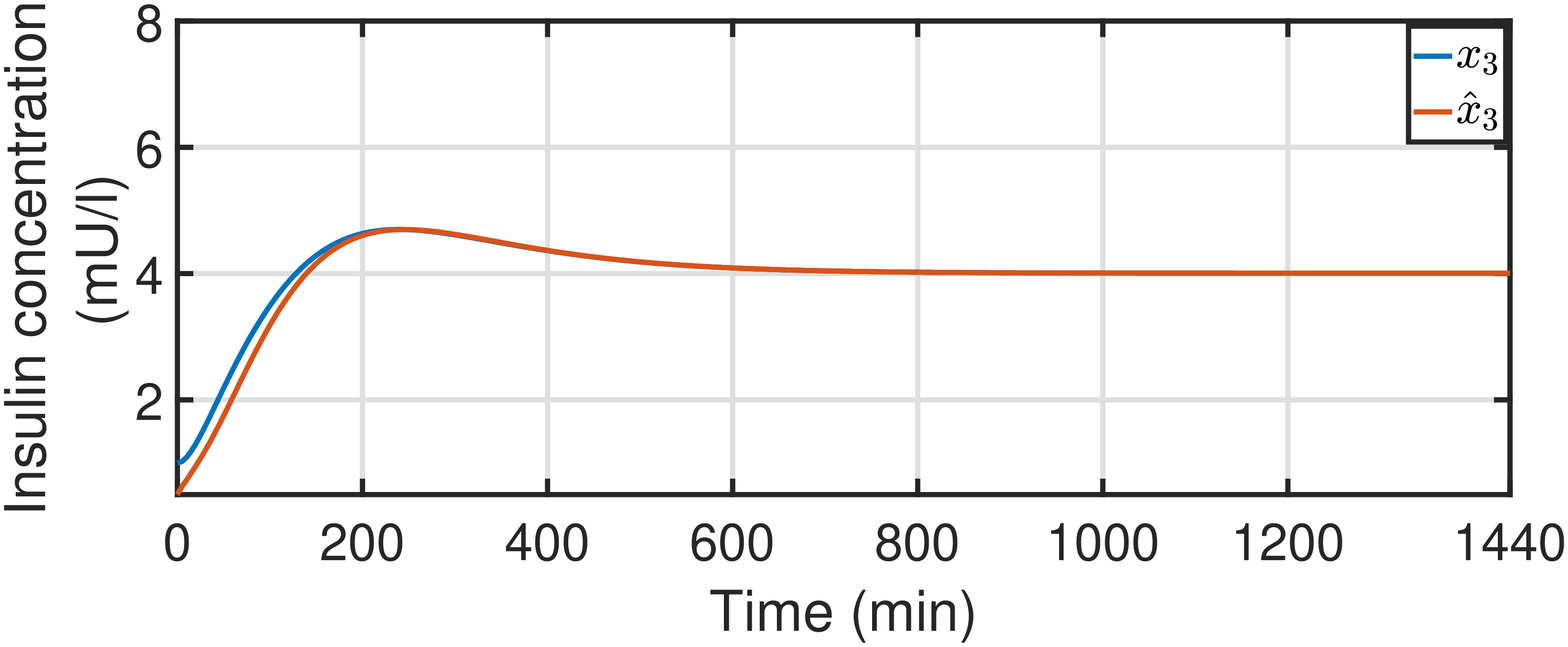}
}
\caption{Results for the estimation of plasma glucose concentration and plasma insulin concentration of (a) subject 1,(b) subject 2 and (c) subject 3.}
\label{obs_est_plot}
\end{figure*}

\subsubsection{Estimation results:} The performance of the observer is evaluated in terms of estimation of the plasma glucose concentration and the plasma insulin concentration for three T1D subjects. A hypothetical 24 h scenario is considered with the assumption that the T1D subjects receive a single meal of 70 g carbohydrate at t=10 min followed by a fasting period. The parameters of the T1D subjects are referred from Table \ref{tab1}. The initial conditions are chosen as $x_1$=120 mg/dl, $x_2=0.01~{min^{-1}}$, $x_3=1$ mU/l and $x_4=1$ mU/l. Observer gains are chosen by satisfying the inequalities in \eqref{ob_ine} while considering minimum glucose concentration ($x_1=0$ mg/dl). The corresponding results of the state estimation are depicted in Fig. \ref{obs_est_plot}. It can be observed that the convergence of the estimate $\hat x_1$ to the true state $x_1$ is very fast for all the cases. Although $\hat x_3$ converges to the actual state $x_3$ in a finite time, the convergence in case of T1D Subject 1 is comparatively slower than Subjects 2 and 3. Here, it is verified that the observer is capable of estimating both the glucose and insulin concentrations. Next, the design of the control law based on the contraction analysis is presented. 

\subsection{Controller Design for Artificial Pancreas}
The most important feature that an effective control law should possess for its implementation in APS is that it should avoid severe hypoglycemia ($x_1<50$ mg/dl) at any circumstance. Additionally, it should ensure an automated continuous insulin delivery to minimize the postprandial glucose excursions and minimize the risks of hyperglycemia ($x_1>180$ mg/dl). The blood glucose should be below 180 mg/dl within 1 h after meal intake. All these control objectives need to be achieved by utilizing the state estimates provided by the observer and in the presence of parametric uncertainty and exogenous meal disturbances.\par 

The convergence of the estimated states to the actual states has been already established in Theorem \ref{theorem-1}. In the next step, the deviated dynamics in \eqref{dev_sys} is considered to formulate the control problem as a regulation problem. By ensuring the deviated state $x_d$ converge to zero, the actual states converge to the equilibrium point. To achieve this objective, the a feedback proportional control law is chosen as
\begin{equation}\label{control}
    u=K\hat{x}_d
\end{equation}
where $K=[k_1\ k_2\ k_3\ k_4]$ is the controller gain matrix and $\hat x_d=[\hat x_{1d}~\hat x_{2d}~\hat x_{3d}~\hat x_{4d}]^T$ is the estimate of the deviated state, $x_d=[x_{1d},x_{2d},x_{3d},x_{4d}]^T$. The information of $\hat{x}_d$ can be extracted from the transformation to the estimated states $\hat{x}$ as introduced in \eqref{transform} as
\begin{equation}\label{xd_hat}
    [\hat x_{1d}~\hat x_{2d}~\hat x_{3d}~\hat x_{4d}]^T=[\hat x_1~\hat x_2~\hat x_3~\hat x_4]^T-[EGP/p_1~0~0~0]^T.
\end{equation}
The estimation error can be obtained as
\begin{equation}\label{err}
    e=x_d-\hat x_d
\end{equation}
By substituting the control law in \eqref{control}, the closed loop deviated dynamics \eqref{dev_sys} can be re-written in compact form as
\begin{equation}\label{con1}
\begin{split}
    \dot{x}_d=&f(x_d)+BK\hat x_d
    \end{split}
\end{equation}
which can be re-written as
\begin{equation}\label{con2}
\begin{split}
    \dot{x}_d=&f(x_d)+BK(x_d-e)\\
             =&f(x_d)+BKx_d-BKe\\
             =&f(x_d)_{CL}+d(t)\\
    \end{split}
\end{equation}\\
where $f(x_d)_{CL}=f(x_d)+BKx_d$ and $d(t)=-BKe$. In general, for a state-feedback based control law, $u=K\hat x_d$, the controller gain matrix, $K$ needs to be designed.The controller gains, $k_i,i=1,\dots,4$ in \eqref{control} need to be selected to ensure the stability of closed loop system that includes the observer states in \eqref{ob} and the control law in \eqref{control}. The stability analysis of closed loop system using contraction theory is presented in the form of following theorem.  
\begin{thm}
Consider the closed loop system in \eqref{con2} with $d(t)=-BKe(t)$ be the vector function satisfying $|d(t)|\leq c_1e^{-c_2t}$, there exist a controller gain matrix, $K=[k_1,k_2,k_3,k_4]^T$ dictated by the following relations in \eqref{con_ine} such that $f(x_d)_{CL}$ is contracting. 
\begin{equation}\label{con_ine}
\begin{array}{lll}
-p_1-x_{2d}+k_{1} p_{6}&<&0 \\ 
-p_2+k_{2} p_{6}+(EGP/p_1)-x_{1d}&<&0 \\
-p_4+p_3+k_3 p_6&<&0 \\ 
k_4p_6-p_5+p_4&<&0,
\end{array}
\end{equation}
The states of closed loop system (controller-observer put together) will exponentially converge to the equilibrium. Further, it will imply that the following relation will be satisfied 
\begin{equation}\label{expo_decay}
    |x_d|< e^{-\alpha_1 t}(\alpha_2 t+|x_{d0}|)
\end{equation}
where $\alpha_1, \alpha_2>0$ are two positive constants that depend upon the initial conditions and the minimum contraction rates of the observer and the system.
\end{thm}
\begin{proof}
The proof of the above theorem is straightforward as stated below. Consider the closed loop dynamics as mentioned in \eqref{con2}\\
where $f(x_d)=\begin{bmatrix}
    -p_1x_{1d} + (EGP/p_1) x_{2d}-x_{1d}x_{2d} \\
    -p_2 x_{2d} + p_3 x_{3d} \\
   -p_4 x_{3d} + p_4 x_{4d}\\
    -p_5 x_{4d} 
    \end{bmatrix}$, $B=\begin{bmatrix}
    0\\
    0 \\
    0\\
    p_6 
    \end{bmatrix}$ \\
    then $f(x_d)_{CL}$ becomes 
    \begin{equation*}\label{nominal}
    f(x_d)_{CL}=\begin{bmatrix}
    -p_1x_{1d} + (EGP/p_1) x_{2d}-x_{1d}x_{2d} \\
    -p_2 x_{2d} + p_3 x_{3d} \\
   -p_4 x_{3d} + p_4 x_{4d}\\
    -p_5 x_{4d}
    \end{bmatrix}+\begin{bmatrix}
    0\\
    0 \\
    0\\
    p_6 
    \end{bmatrix}{\begin{bmatrix}
    k_1\\
    k_2\\
    k_3\\
    k_4 
    \end{bmatrix}}^T\begin{bmatrix}
    x_{1d}\\
    x_{2d}\\
    x_{3d}\\
    x_{4d}
    \end{bmatrix}
    \end{equation*}\\
    In compact form,
    \begin{equation}\label{nominal}
    f(x_d)_{CL}=\begin{bmatrix}
    -p_1x_{1d} + (EGP/p_1) x_{2d}-x_{1d}x_{2d} \\
    -p_2 x_{2d} + p_3 x_{3d} \\
   -p_4 x_{3d} + p_4 x_{4d}\\
    -p_5 x_{4d}+p_6 k_1 x_{1d}+p_6k_2x_{2d}+ p_6k_3x_{3d}+p_6k_4x_{4d} 
    \end{bmatrix}
    \end{equation}\\
    The differential dynamics of closed loop system is in the form,
    \begin{equation}
        \delta \dot{x}_{d}=J_c \delta x_d
    \end{equation}\\
    where 
      $J_c=
    \begin{bmatrix}
    -p_1-x_{2d} & (EGP/p_1)-x_{1d} & 0 & 0\\
    0 & -p_2 & p_3 & 0\\
    0 & 0 & -p_4 & p_4\\
    k_1 p_6 & k_2 p_6 & k_3 p_6 & k_4 p_6-p_5
    \end{bmatrix}$\\
    
     Similar to the observer design in \ref{observer}, the controller gain $K$ can be selected in the framework of partial contraction stated in Lemma 1. Therefore, designed $K$ will ensure the closed loop system $(f(x_d)_{nominal})$ to be contracting to equilibrium and $e(t)\rightarrow 0$ exponentially as the observer dynamics achieves exponential convergence. Here, the term $d(t)=-BKe(t)$ can be treated as a decaying perturbation to the nominal system. Using Lemma 2, exponential convergence of states  to the equilibrium is guaranteed as given in \eqref{expo_decay}. 
\end{proof}


For the nominal system to be contracting, we need to design $K$ such that the matrix measure of $J_c$ is negative. Using matrix measure, the design conditions can be obtained as \eqref{con_ine}. 


\begin{rmk}
It is important to highlight that the observer and controller design can be carried out separately. The reason behind this is that the system comprising of the cascade interconnection of two individually contracting systems, is also contracting \cite{contraction}. Hence, this work is less restrictive than the other works on observer-based controller in \cite{guarcost,ijs,aem}.
\end{rmk}
\section{Results}
For evaluating the efficacy of the proposed observer-based control technique, two broad simulation categories are undertaken as case studies. The first category represents a realistic daily scenario of T1D patients, where three meals of 75 g carbohydrates, a representative of breakfast, lunch, and dinner, are taken into consideration. The corresponding observer and controller gains are chosen based on the inequalities in \eqref{ob_ine} and \eqref{con_ine}, respectively. 

\subsection{Scenario 1: A single day, three meal scenario with nominal parameters}
\label{scenario1}
The T1D model, as presented in \eqref{sys}, is being considered to represent the T1DPs. This virtual simulation scenario of 24 h (equivalent to 1440 min) is designed to investigate the performance of the controller in addressing inter-patient variability. So three virtual T1DPs, namely, Subjects 1, 3, and 5, are taken into consideration and the corresponding parameters are adopted from \cite{ivp} as mentioned in Table \ref{tab1}. Three meals containing an equal amount of carbohydrate of 75 g are provided as breakfast, lunch, and dinner at t=10, 360, and 720 min, respectievely. It is assumed that the simulations start from a safe blood glucose level, $x_1$ of 120 mg/dl and initial conditions $x_2=0.01~{min^{-1}}$, $x_3=1$ mU/l and $x_4=1$ mU/l.

\begin{figure}[h!]
\centering
\subfloat[]{
  \includegraphics[width=70mm,height=3cm]{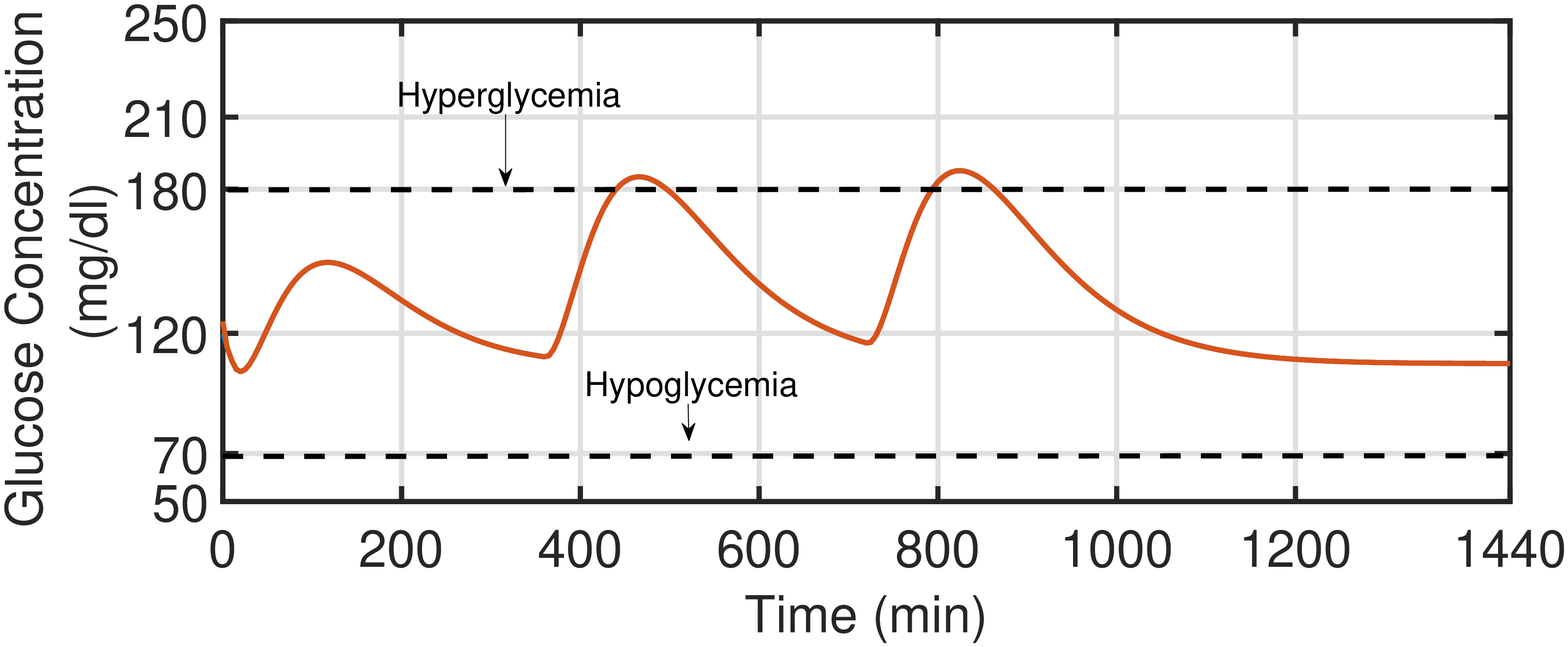}
}
\\
\subfloat[]{
  \includegraphics[width=70mm,height=3cm]{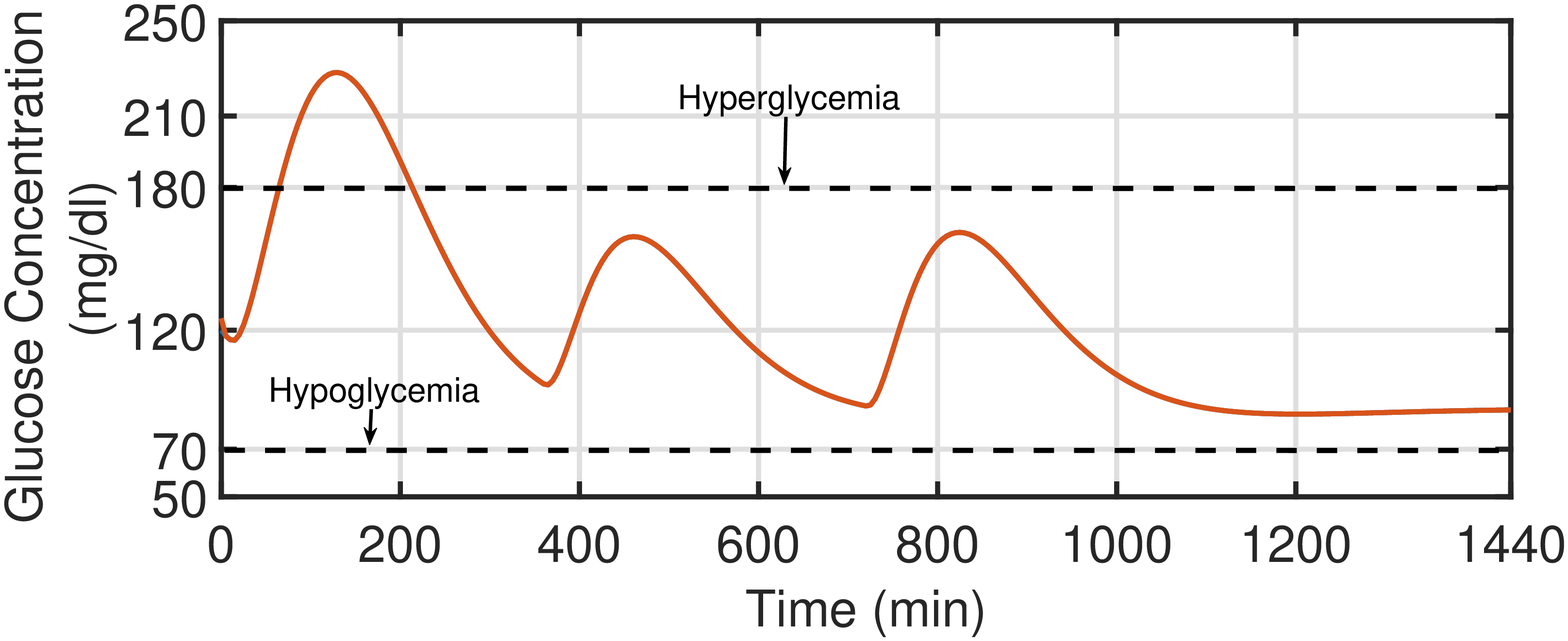}
}
\\
\subfloat[]{
  \includegraphics[width=70mm,height=3cm]{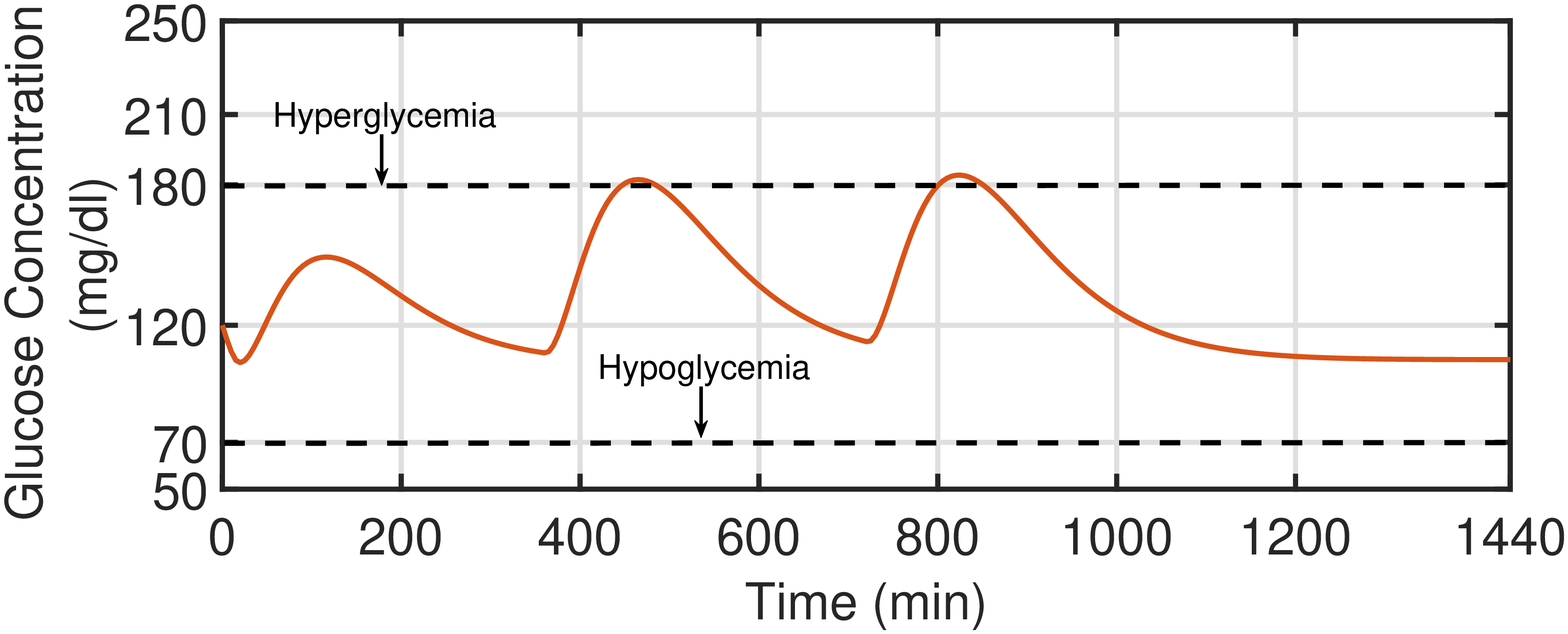}
}
\caption{Plasma glucose concentration trajectories of (a) subject 1,(b) subject 2 and (c) subject 3 under the proposed observer-based feedback control law.}
\label{g_plot}
\end{figure}

\begin{figure}[h!]
\centering
\subfloat[]{
  \includegraphics[width=70mm,height=3cm]{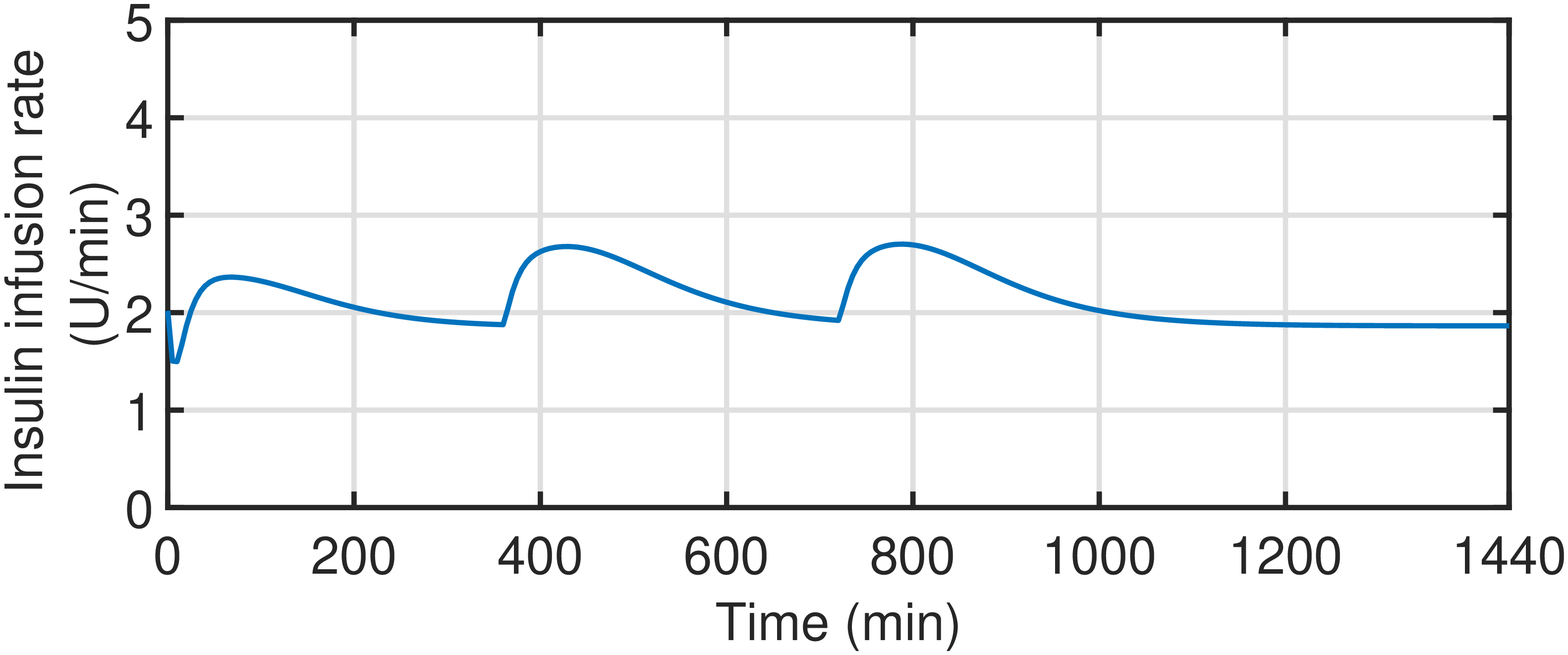}
}
\\
\subfloat[]{
  \includegraphics[width=70mm,height=3cm]{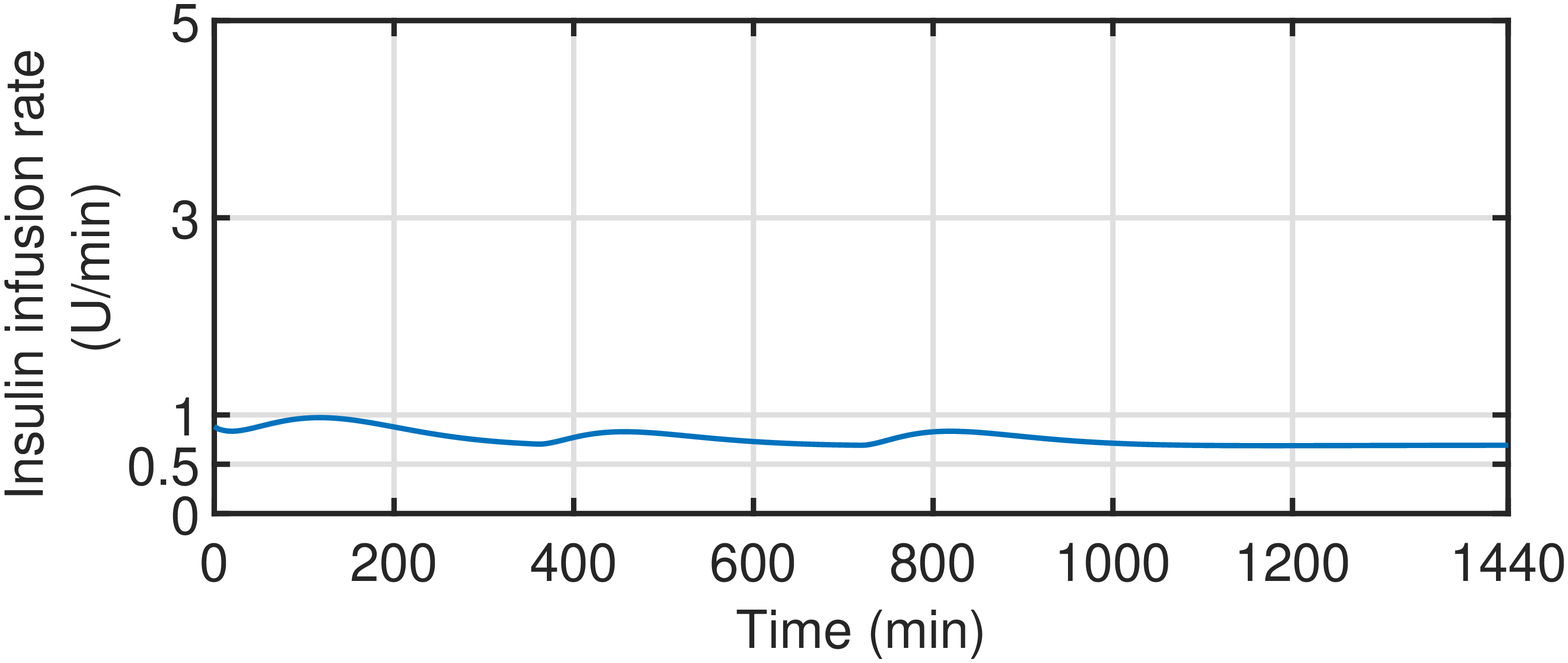}
}
\\
\subfloat[]{
  \includegraphics[width=70mm,height=3cm]{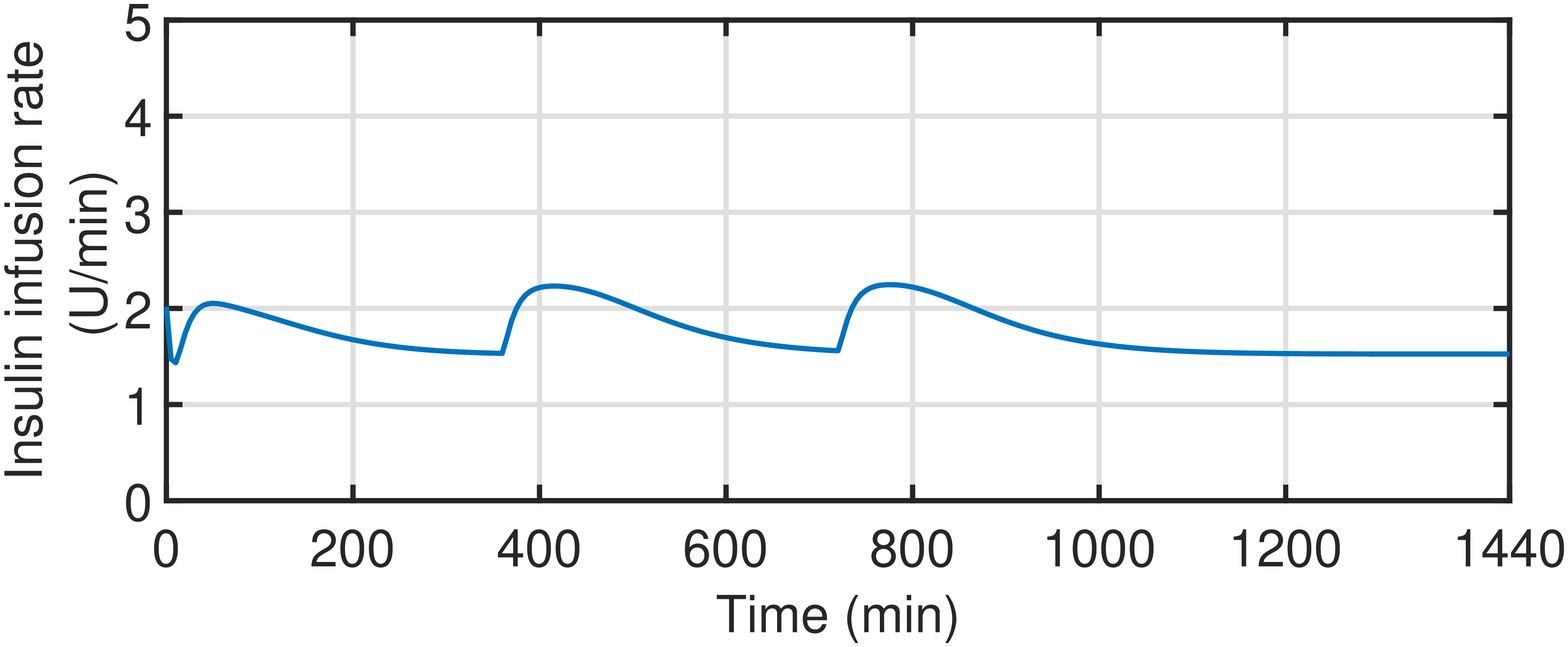}
}
\caption{Insulin infusion rate as suggested by the proposed observer-based feedback control law for (a) subject 1,(b) subject 2 and (c) subject 3.}
\label{u_plot}
\end{figure}

As illustrated in Fig. \ref{g_plot}, the glucose concentration is effectively regulated within the euglycemic range ($x_1\in [70-180]$ mg/dl) for all the three T1D subjects taken under consideration. In the case of Subject 3, there is a glucose excursion above 180 mg/dl in the postprandial period after the first meal. However, the glucose is brought below the hyperglycemic clamp ($x_1=180$ mg/dl) within two hours. During the whole simulation period, no instance of hypoglycemic ($X-1<70$ mg/dl) episodes are recorded for all the cases. It validates the ability of the proposed control algorithm in the nominal conditions. The corresponding control signals in terms of the exogenous subcutaneous insulin infusion rate are shown in Fig. \ref{u_plot}. It can be observed that an almost constant insulin delivery rate is maintained for the 24 h simulations. It is important to highlight that there is a transient increase in insulin delivery rate when there is a transient growth in glucose level after meal intakes. This feature of the control law is highly desirable for minimising the risks of postprandial hyperglycemia. The smooth profile of the exogenous insulin delivery rate makes the controller action less aggressive, which can help avoid late hypoglycemia induced by excessive insulin infusion if any.  

\begin{table*}[h!]
\centering
\caption{Simulation setting for Scenario 2A-2D for T1D Subjects 1, 3 and 5.}
\label{tab_protocol}
\begin{tabular}{|l|l|l|l|l|l|l|l|l|l|}
\hline
\multirow{2}{*}{Scenario} & \multicolumn{5}{|l|}{Parameters}                            & \multicolumn{4}{|l|}{Initial conditions}         \\
                          & $p_1$     & $p_2$     & $p_3$     & $p_4$     & $p_5$     & $x_{10}$     & $x_{20}$ & $x_{30}$   & $x_{40}$ \\ \hline
2A                        & Nominal   & Nominal   & $\pm30\%$ & Nominal   & Nominal   & 120          & 0.1      & 1          & 1        \\ \hline
2B                        & Nominal   & Nominal   & Nominal   & $\pm30\%$ & $\pm30\%$ & 120          & 0.1      & 1          & 1        \\ \hline
2C                        & $\pm30\%$ & $\pm30\%$ & $\pm30\%$ & $\pm30\%$ & $\pm30\%$ & 120          & 0.1      & 1          & 1        \\ \hline
2D                        & $\pm30\%$ & $\pm30\%$ & $\pm30\%$ & $\pm30\%$ & $\pm30\%$ & {[}80 140{]} & 0.01     & {[}0 10{]} & 1  \\ \hline
\end{tabular}
\end{table*}

\subsection{Scenario 2: A single day, 3 meal scenario with intra-patient variability}\label{scenario2}
The main focus of the current scenario is to investigate the effect of bounded parametric uncertainty on the closed-loop performance. Four different simulation settings are further considered that take explicit variability in insulin sensitivity (IS), insulin pharmacokinetics and pharmacodynamics, and the initial conditions of glucose and insulin in the body. The observer and the controller feedback parameters are computed only for the T1D patients considering the nominal parameter values. For testing the robustness of the proposed technique, these feedback parameters are kept constant for each subject. One hundred numerical simulations are carried out for each of the T1D subjects with random parametric perturbations. The different protocols of this scenario are provided in Table \ref{tab_protocol}. The details about the different simulation settings are presented below.

\subsubsection{Scenario 2A: Variability in insulin sensitivity (IS)}
IS can be critical physiological parameters responsible for variability in the glucose stabilization in T1DPs. The parameter, $p_3$ that is directly related the IS and is assumed to vary in a range of $\pm30\%$ of its nominal value. A total of 100 simulations are carried out with randomly varying $p_3$ in $\pm30\%$ uncertainty bound specified in Table \ref{tab_protocol}.\par

From Table \ref{tab_time}, it can be inferred that the proposed control law is highly efficient in dealing with the uncertainty in IS as noted in Table \ref{tab_protocol}. For all the cases, percentage of time spent in euglycemic range ($x_1\in [70-180]$ mg/dl) is more than 84\% and that in hyperglycemic range ($x_1>180$ mg/dl) is less than 16\%. No hypoglycemia ($x_1<70$ mg/dl) or severe hypoglycemia ($x_1<50$ mg/dl) observed from the simulation results with variable IS. 

\begin{table}[h!] \small%
\centering
\caption{Percentage of time spent by the the virtual patients in the euglycemic ($x_1\in [70-180]$ mg/dl), hyperglycemic ($x_1>180$ mg/dl) and hypoglycemic ($x_1\in [50-70]$ mg/dl).}
\label{tab_time}
\begin{tabular}{|p{1.2cm}|p{1.2cm}|p{1.2cm}|p{1.2cm}|p{1.2cm}| p{1.2cm}|}
\hline
Subject & Scenario No. & Time Spent in euglycemic range (\%)  & Time spent in hyperglycemic range (\%) &Time spent in hypoglycemic range (\%) \\
\hline
\cline{2-3}\cline{3-3} \cline{4-4}\cline{5-5}
Subject 1 & 2A & 84.08 &15.92 & 0  \\
\cline{2-3}\cline{3-3} \cline{4-4}\cline{5-5}
                 & 2B & 100 & 0 & 0 \\
\cline{2-3}\cline{3-3} \cline{4-4}\cline{5-5}
                 & 2C & 77.85 & 22.15 & 0 \\
\cline{2-3}\cline{3-3} \cline{4-4}\cline{5-5}
                 & 2D & 46.37 & 53.63 & 0 \\
\hline
\cline{2-3}\cline{3-3} \cline{4-4}\cline{5-5}
Subject 3 & 2A & 87.20 &12.80 & 0  \\
\cline{2-3}\cline{3-3} \cline{4-4}\cline{5-5}
                 & 2B & 69.55 & 30.45 & 0 \\
\cline{2-3}\cline{3-3} \cline{4-4}\cline{5-5}
                 & 2C & 78.89 & 21.11 & 0 \\
\cline{2-3}\cline{3-3} \cline{4-4}\cline{5-5}
                 & 2D & 59.86 & 40.14 & 0 \\
\hline
\cline{2-3}\cline{3-3} \cline{4-4}\cline{5-5}
Subject 5 & 2A & 92.73 &7.27 & 0  \\
\cline{2-3}\cline{3-3} \cline{4-4}\cline{5-5}
                 & 2B & 92.04 & 7.96 & 0 \\
\cline{2-3}\cline{3-3} \cline{4-4}\cline{5-5}
                 & 2C & 76.82 & 23.18 & 0 \\
\cline{2-3}\cline{3-3} \cline{4-4}\cline{5-5}
                 & 2D & 42.21 & 5.80 & 51.90 \\
\hline
\end{tabular}
\end{table}

\subsubsection{Scenario 2B: Variability in insulin pharmacokinetics and pharmacodynamics}
The variation in the insulin absorption dynamics due to the uncertainty in the parameters, $p_j,j=4,5$, is also significant. A variation of $\pm30\%$ about nominal values is assumed to carry out 100 numerical simulations for the T1D subjects. During each simulation, the parameters are randomly initialized from the uncertainty bound is given in Table \ref{tab_protocol}.\par

From Table \ref{tab_time}, it can be concluded that the designed controller possesses the right capability to handle the issues related to the uncertainty in the subcutaneous insulin absorption. The percentage of time spent in the euglycemic range is 84\% and 100\% for Subjects 1 and 5, respectively. While in the case of Subject 3, it is around 70\%. The time spent in hypoglycemia ($x_1<70$ mg/dl) is 0\% as provided in Table \ref{tab_time}. 

\subsubsection{Scenario 2C: Variability in all model parameters}
Next, an uncertain case is considered where all the system parameters, $p_i,i=1,...,5$, are varied in the range $\pm30\%$. The objective of this scenario is to determine the effectiveness of the proposed technique in ensuring an efficient glucose regulation.\par

In this case study, Subjects 1, 3, and 5 spent almost 76\% of the total time in the euglycemic range. The corresponding time spent above 180 mg/dl is below 22\% for all the cases. Still, it is a significant improvement of glycemic control under a wide range of variation in all the system parameters. Complete avoidance of hypoglycemic instances has been recorded in this case, as presented in Table \ref{tab_time}.

\begin{figure}[!h]
		\centering
		\includegraphics[width=1.0\linewidth]{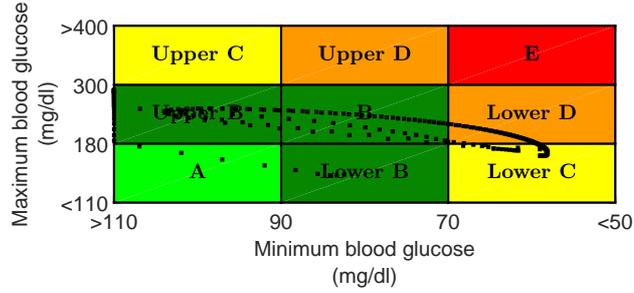}
		\caption{CVGA for parametric variability of $\pm30\%$}
		\label{cvga}\vspace{-0.4cm}
\end{figure}

\subsubsection{Scenario 2D: Variability in initial values of glucose and insulin}
To further complicate the situation, it is assumed that the initial glucose concentration varies within $x_1\in [80-140]$, and the initial insulin concentration vary within $x_3\in [0-10]$ mU/l. Like in the previous case, the parameters $p_i,i=1,...,5$ are assumed to vary in the range of $\pm30\%$. The details of the protocol can be referred to from Table \ref{tab_protocol}. A control variability grid analysis (CVGA) \cite{NATH20197} is performed to evaluate the efficacy of the proposed observer-based control technique as illustrated in Fig. \ref{cvga}.\par

As depicted in Fig. \ref{cvga}, all the black dots are confined to Grid A, B, Upper B, Lower B, Lower D, and Lower C in the CVGA plot. There are a total of 300 dots (100 dots per subject with random parameters and initial conditions), which correspond to 300 random simulations in total. Each of these dots is mapped to a particular co-ordinate in the CVGA plot (depending upon the maximum and minimum glucose level achieved during the simulation). Severe hypoglycemia, i.e., glucose level below 50 mg/dl, is successfully avoided in all of these random numerical simulations. Approximately, 75\% of the dots are confined to the safe grids (shown in green). Although 15\% (approx) are in Grid Lower D and Grid Lower C, they are still above 50 mg/dl safety constraint. It is worth mentioning at this stage that the closed-loop performance of the proposed algorithm is safe and effective. It is also important to note that the average time spent in the euglycemic range (70-180 mg/dl) by the three subjects is found to be approximately equal to 77\%, 73\% and 76\%, respectively, for scenarios 2A-2D. It validates the effectiveness of the proposed control algorithm in achieving improved glycemic control.

\subsubsection{Quantitative performance evaluation of the closed loop simulations}
The two most important features of the analysis of any glucose regulation algorithm are safety and efficacy. The former is investigated extensively in the previous simulation scenarios. The later is being investigated in the current section. Several important performance indices that are critical to evaluate the closed-loop glucose regulation include the Low Glucose Variability Index (LBGI), High Glucose Variability Index (HBGI), Coefficient of Variation (CoV), mean blood glucose and HbA1c \cite{kovatchev2005quantifying}. The LBGI and HBGI for the Scenarios 2A-2D are presented as bar plots, as illustrated in Fig. \ref{lbgi}. The rest of the indices are computed, and the corresponding results are summarised in Table \ref{tab_stat}. 

\begin{table}[h!] \small%
\centering
\caption{Statistical analysis of the closed loop results obtained during Scenario 2A-2D.}
\label{tab_stat}
\begin{tabular}{|p{1.2cm}|p{1.2cm}|p{1.2cm}|p{1.2cm}|p{1.2cm}| p{1.2cm}|}
\hline
Subject & Scenario No. & Mean glucose (mg/dl)  & CoV (\%) & HbA1c (\%) \\
\hline
\cline{2-3}\cline{3-3} \cline{4-4}\cline{5-5}
Subject 1 & 2A & 148.69 & 0.0222 & 6.8081  \\
\cline{2-3}\cline{3-3} \cline{4-4}\cline{5-5}
                 & 2B & 136.48 & 0.0229 & 6.3827 \\
\cline{2-3}\cline{3-3} \cline{4-4}\cline{5-5}
                 & 2C & 152.31 & 0.1325 & 6.9342 \\
\cline{2-3}\cline{3-3} \cline{4-4}\cline{5-5}
                 & 2D & 148.99 & 0.1195 & 6.6791 \\
\hline
\cline{2-3}\cline{3-3} \cline{4-4}\cline{5-5}
Subject 3 & 2A & 144.98 & 0.1087 & 6.6791  \\
\cline{2-3}\cline{3-3} \cline{4-4}\cline{5-5}
                 & 2B & 145.56 & 0.9910 & 6.6991 \\
\cline{2-3}\cline{3-3} \cline{4-4}\cline{5-5}
                 & 2C & 139.96 & 0.2022 & 6.5039 \\
\cline{2-3}\cline{3-3} \cline{4-4}\cline{5-5}
                 & 2D & 140.69 & 0.1945 & 6.5294 \\
\hline
\cline{2-3}\cline{3-3} \cline{4-4}\cline{5-5}
Subject 5 & 2A & 114.23 & 0.1113 & 5.6076  \\
\cline{2-3}\cline{3-3} \cline{4-4}\cline{5-5}
                 & 2B & 111.26 & 0.0811 & 5.5040 \\
\cline{2-3}\cline{3-3} \cline{4-4}\cline{5-5}
                 & 2C & 112.75 & 0.1884 & 5.5559 \\
\cline{2-3}\cline{3-3} \cline{4-4}\cline{5-5}
                 & 2D & 110.90 & 0.1891 & 5.4916 \\
\hline
\end{tabular}
\end{table}

\begin{figure}[h!]
\centering
\subfloat[]{
  \includegraphics[width=70mm,height=3cm]{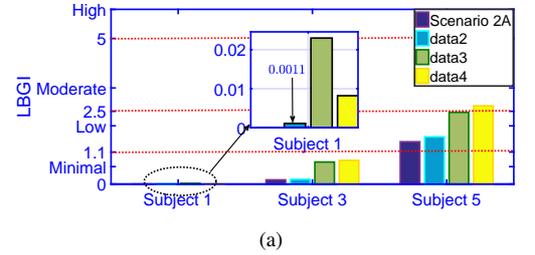}
}
\\
\subfloat[]{
  \includegraphics[width=70mm,height=3cm]{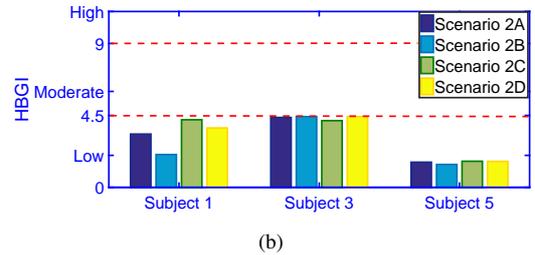}
}
\caption{LBGI and HBGI for the (a) subject 1,(b) subject 2 and (c) subject 3 for Scenarios 2A-2D.}
\label{lbgi}
\end{figure}

The statistical results of the closed-loop simulations of Scenario 2A-2D are thoroughly summarized in Table \ref{tab_stat}. The mean glucose level of 136.48-152.31 mg/dl, 139.96-145.56 mg/dl, and 110.90-114.23 mg/dl are maintained for all the virtual T1D subjects under large perturbations. The quality of glycemic regulation obtained by the proposed observer-based controller is reflected with around  HbA1c of 6.5\%, which is one of the main aims of any standard diabetes treatment \cite{schnell2017impact}. The above-mentioned facts can also be verified from the attainment of a low value of Coefficient of Variability (CoV). In all the cases, CoV is below 2\%, which is an indicator of low glucose variability under the feedback action. The LBGI and HBGI predict the risk of hypoglycemia and hyperglycemia, respectively. The risk analysis that is directly related to the reliability of any glucose controller depends on these two factors. The LBGI and HBGI for Subjects 1, 3, and 5 are presented in a bar plot in Fig. \ref{lbgi}. It can be easily observed from Fig. \ref{lbgi}(a) that the LBGI for Subject 1 and 3 are in the minimal LBGI region (LBGI$<1.1$), which is significant in avoiding the risks of hypoglycemia. For T1D subject 5, LBGI is low. The corresponding HBGI are also in the Low LBGI region, as illustrated in Fig. \ref{lbgi}(b). Hence, the results validate that the risks related to the postprandial hyperglycemia are low under the designed control action.

\section{Conclusion}
The concept of contraction analysis is utilized in the present work for designing an observer and a controller to achieve a tight glycemic control in the presence of intra-patient variability. The proposed method is shown to be effective in handling intra-patient variability resulting from different sources of uncertainty.  Extensive numerical simulations are carried out to evaluate the performance of the proposed technique for realistic scenarios. The postprandial hyperglycemic events are significantly minimized, and hypoglycemia is avoided. The efficacy of the control algorithm is evaluated through statistical analysis.  The simple structure of the observer and control law makes it a desirable candidate for Artificial Pancreas. The control scheme can be extended to more complicated models like the UVa Padova model, Hovorka model, etc., that consider a detailed dynamics where the effect of glucagon and free fatty acids. The design philosophy can be extended for nonlinear systems with a quantized and sampled output, which represents the actual glucose measurements done by the CGM devices in practice.

\section*{Conflicts of interest}
The authors declare no conflicts of interest.
\bibliographystyle{elsarticle-num}
\bibliography{reference}

\end{document}